\documentclass[aps,pra,reprint,superscriptaddress,amsmath,amssymb,floatfix]{revtex4-2}

\usepackage{amsthm}
\usepackage{graphicx}
\usepackage[hidelinks]{hyperref}
\theoremstyle{plain}
\newtheorem{theorem}{Theorem}
\newtheorem{corollary}{Corollary}

\begin{document}
\raggedbottom

\title{Enhancing Quantum Metrology with High-order Fisher Information and Experiments}

\author{Xin-Zhu Liu}
\affiliation{School of Information Science and Technology, Southwest Jiaotong University, Chengdu 610031, China}
\author{Jun-Li Jiang}
\affiliation{School of Information Science and Technology, Southwest Jiaotong University, Chengdu 610031, China}
\author{Yan-Han Yang}
\affiliation{School of Information Science and Technology, Southwest Jiaotong University, Chengdu 610031, China}
\author{Li-Ming Zhao}
\affiliation{School of Information Science and Technology, Southwest Jiaotong University, Chengdu 610031, China}
\author{Xue Yang}
\affiliation{School of Information Science and Technology, Southwest Jiaotong University, Chengdu 610031, China}
\author{Shao-Ming Fei}
\affiliation{Max-Planck-Institute for Mathematics in the Sciences, 04103 Leipzig, Germany}
\author{Ming-Xing Luo}
\email{mxluo@swjtu.edu.cn}
\affiliation{School of Information Science and Technology, Southwest Jiaotong University, Chengdu 610031, China}
\affiliation{Hefei National Laboratory, University of Science and Technology of China, Hefei, Anhui 230088, China}

\begin{abstract}
Fisher information plays a central role in statistics and quantum metrology, providing the basis for the celebrated Cram\'er--Rao bound. In this work, we introduce a new information measure based on higher-order Fisher information and show that it naturally leads to a generalized uncertainty relation for parameter estimation, which can be regarded as an extension of the Cram\'er--Rao bound.
As an application, we analyze the case of quantum phase estimation and compare our theoretical bounds with the well-known established hierarchical bounds. Finally, we experimentally validate the proposed framework using a photonic platform.
\end{abstract}

\maketitle

\section{Introduction}
\label{sec:intro}

Fisher information (FI) characterizes the sensitivity of a statistical
model and determines the Cram\'er--Rao bound (CRB) on estimation variance
under the usual regularity and unbiasedness conditions
\cite{Cramer1999,Rao1992}.
Its quantum counterpart, the quantum Cram\'er--Rao bound (QCRB), is
governed by the symmetric-logarithmic-derivative (SLD) quantum Fisher
information and provides a fundamental precision benchmark optimized
over quantum measurements \cite{Helstrom1969,Braunstein1994}.
Developments in distributed phase estimation
\cite{liu2021,Kim2024Distributed,liu2024},
adaptive methods \cite{valeri2020,Minati2026Adaptive},
and photonic implementations \cite{polino2020,Cimini2024Variational}
have extended the practical scope of quantum metrology, with applications
ranging from interferometry to quantum sensing
\cite{Giovannetti2006,Degen2017,Pezze2018,Montenegro2025Review,Antonella2016}.

Recent studies connect the structure of quantum states with both
sensitivity and measurement compatibility
\cite{Albarelli2026Measurement}.
Multiparameter critical metrology investigates how quantum criticality
enhances sensitivity and improves the compatibility of optimal
measurements for different parameters
\cite{DiFresco2022Multiparameter,DiFresco2025CriticalReview}.
The scaling of quantum Fisher information also connects metrology
with multipartite entanglement in measurement-induced phase
transitions \cite{DiFresco2024Measurement}.
For fermionic Gaussian states, explicit SLD expressions enable
analytical evaluations of quantum Fisher information
\cite{Carollo2018FermionicSLD}.
These developments highlight the value of relating quantum-state
structure to computable precision constraints.

Higher-order derivatives and finite parameter displacements lead to
frequentist hierarchies, including the quantum Bhattacharyya (QBh),
Hammersley--Chapman--Robbins (QHCR), Abel (QAB) and Barankin bounds
\cite{Bhattacharyya1946,Hammersley1950,Chapman1951,
Abel1993,Barankin1949,Gessner2023}.
The first-order QBh bound coincides with the QCRB.

Finite-sample analyses and higher-order corrections also refine
the characterization of estimation uncertainty
\cite{Gebhart2024FewMeasurements,Hervas2025HigherOrder,
Meyer2025FiniteSample}.

Square-root representations also provide a geometric route to
higher-order uncertainty relations \cite{BelfieldBrody2021}.
This motivates a quantitative comparison between explicit square-root
curvature constraints and the derivative hierarchy, including the
conditions under which curvature supplies a stronger variance constraint.

Here we propose second-order information from the square-root
representations of classical likelihoods and quantum states
(Fig.~\ref{fig1}). This construction links the proposed second-order
information measure to estimation uncertainty, yielding variance bounds
for classical and quantum parameter estimation.
We characterize its spectral structure, composition under independent
copies, and multiparameter extension, providing a systematic account of
statistical sensitivity beyond first-order Fisher information.
Centering further strengthens the resulting variance constraints
analytically.
We investigate their relation to widely used Cram\'er--Rao, Bhattacharyya,
Hammersley--Chapman--Robbins, and Abel bounds, encompassing approaches
based on local sensitivity, higher-order derivatives, and finite
parameter displacements.
Together, these results establish the proposed information measure and its associated bounds as tools for obtaining sharper uncertainty
constraints in quantum metrology.
We apply the framework to quantum phase estimation and compare the
resulting constraints with established frequentist bounds.
The setting is further investigated on a photonic platform.

\begin{figure}[!t]
    \centering
    \includegraphics[width=\linewidth]{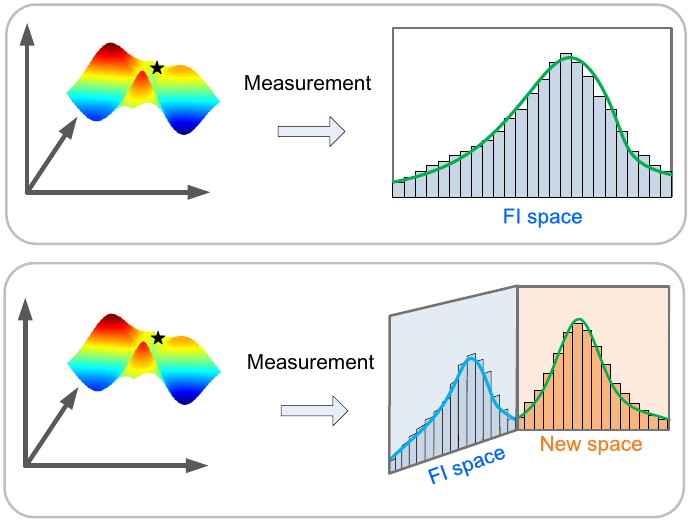}
    \caption{Schematic metrology using Fisher information alone (top) or with higher-order information (bottom). Estimation errors of parameterized unknowns are characterized using feature spaces defined by Fisher information and second-order information. }
   \label{fig1}
\end{figure}

\section{Metrology with high-order information}
\label{sec:classical-high-order}
Consider a statistical model $p(x|\theta)$ for estimating an unknown
parameter $\theta$. For a fixed estimator $\hat\theta(x)$,
the variance at a working point $\theta_0$ is
\begin{equation}
 \Delta\hat\theta^2=
 \int p(x|\theta_0)
 \bigl[\hat\theta(x)-\langle\hat\theta\rangle_{\theta_0}\bigr]^2dx.
\end{equation}
Here $\langle\hat\theta\rangle_\theta=
\int p(x|\theta)\hat\theta(x)\,dx$.
Under the usual regularity conditions, Fisher information is
\cite{Cramer1999,Rao1992}
\begin{equation}
 I=4\int\bigl(\partial_\theta\sqrt{p(x|\theta)}\bigr)^2dx,
 \qquad \Delta\hat\theta^2\geq\frac1I
 \label{eq:classical-fi-crb}
\end{equation}
for a first-order locally unbiased estimator. We incorporate the
curvature of the square-root likelihood through
\begin{equation}
 I_2=4\int\bigl(\partial_\theta^2\sqrt{p(x|\theta)}\bigr)^2dx.
 \label{Eq-1}
\end{equation}
The exact log-likelihood expression is
\begin{equation}
 I_2=\int p(x|\theta)
 \left[\partial_\theta^2\ln p(x|\theta)
       +\frac12\bigl(\partial_\theta\ln p(x|\theta)\bigr)^2\right]^2dx.
 \label{eq:classical-log-form}
\end{equation}
It combines first- and second-order sensitivities in an explicit
curvature measure.

For a fixed estimator unbiased through second order at $\theta_0$,
the square-root curvature gives the variance constraint
\begin{equation}
 \Delta\hat\theta^2\geq\frac{4\mathbb E_I^2}{I_2},
 \qquad
 \mathbb E_I=\int[\hat\theta(x)-\theta_0]
       \bigl(\partial_\theta\sqrt{p(x|\theta)}\bigr)^2dx.
 \label{CR1}
\end{equation}
All derivatives and information values here are evaluated at $\theta_0$,
and $I_2>0$ is assumed. Unbiasedness through second order means that
the mean equals $\theta_0$ and its first and second derivatives are
$1$ and $0$, respectively. Equation~\eqref{CR1} resolves an
estimator-dependent curvature correlation and complements the first-order
Fisher-information description \cite{stoica1998,Ben-Haim2010,cianchi2013}.
Its proof is given in Appendix~\ref{app:classical-proof}.
The exact relation to Bhattacharyya information is stated
in Sec.~\ref{sec:root-bh-comparison} and derived in
Appendix~\ref{app:bh-comparison}.

\section{Quantum metrology with high-order information}
\label{sec:quantum-high-order}
\subsection{Second-order information and variance constraints}
\label{sec:root-constraints}

Consider a finite-dimensional quantum-state family $\rho_\theta$
whose square root is twice differentiable in Hilbert--Schmidt norm.
Motivated by the square-root representation of classical Fisher
information, we consider a quantum analogue $I^q$ and introduce
the second-order information $I_2^q$:
\begin{align}
 I^q&=4\operatorname{Tr}\!\left[
 \left(\frac{\partial\sqrt{\rho_\theta}}{\partial\theta}\right)^2
 \right],
 \label{eq:root-first-information}\\
 I_2^q&=4\operatorname{Tr}\!\left[
 \left(\frac{\partial^2\sqrt{\rho_\theta}}{\partial\theta^2}\right)^2
 \right].
 \label{2nd-Qfisher}
\end{align}
For a family with a parameter-independent eigenbasis, $I^q$
reduces to classical Fisher information, while $I_2^q$
characterizes second-order statistical sensitivity.
Throughout Secs.~\ref{sec:root-constraints}--\ref{sec:root-bh-comparison}, we write $\rho=\rho_\theta$ and
$R=\sqrt{\rho_\theta}$, with primes denoting parameter derivatives.
All quantities and local calibration conditions are evaluated
at the parameter value $\theta$ under consideration, with the
readout held fixed under differentiation.

The square-root geometry also underlies higher-order uncertainty
relations for unitary mixed-state families \cite{BelfieldBrody2021}.
Detailed evaluation, properties, and multiparameter extensions
of the proposed information are presented in
Appendices~\ref{app:information-properties}
and~\ref{app:multiparameter-information}.

For a POVM $\{E_x\}$ with fixed reported estimates $h(x)$, let
$M=\sum_xh(x)E_x$ be its Hermitian first-moment operator, with
\begin{equation}
 \mu_M(\theta)=\operatorname{Tr}(M\rho_\theta).
 \label{eq:quantum-mean-and-centering}
\end{equation}
The measurement and reported values remain fixed under
differentiation, including after calibration at $\theta$.
The operator variance and curvature correlation are
\begin{align}
 \Delta M^2&=\operatorname{Tr}\!\left[
 \rho\bigl(M-\mu_M(\theta)\mathbb I\bigr)^2\right],
 \nonumber\\
 \mathbb E_{I^q}&=\operatorname{Tr}\!\left[
 \bigl(M-\mu_M(\theta)\mathbb I\bigr)(R')^2\right].
 \label{eq:root-quantities}
\end{align}
For a general POVM, the proposed-estimate variance satisfies
\begin{equation}
 \operatorname{Var}_{\theta}(\hat\theta)
 =\sum_xh(x)^2\operatorname{Tr}(\rho E_x)
       -\mu_M(\theta)^2\geq\Delta M^2,
 \label{eq:povm-variance-chain}
\end{equation}
with equality for a projective measurement of $M$.
At an unbiased working point, this variance equals the
mean-squared estimation error. Consequently, all bounds derived below for $\Delta M^2$
also apply to the estimator variance
$\operatorname{Var}_{\theta}(\hat\theta)$.

Local unbiasedness through second order means
\begin{equation}
 \mu_M(\theta)=\theta,\qquad
 \mu_M'(\theta)=1,\qquad \mu_M''(\theta)=0.
 \label{eq:second-local-unbiased}
\end{equation}
Unbiasedness throughout a neighborhood implies these conditions.

\begin{theorem}[Second-order correlation bound]
\label{thm:second-order-qbound}
For a fixed estimation operator satisfying
Eq.~\eqref{eq:second-local-unbiased}, if $I_2^q>0$, then
\begin{equation}
 \Delta M^2\geq\frac{4(\mathbb E_{I^q})^2}{I_2^q}.
 \label{2fisherCR}
\end{equation}

\end{theorem}

\begin{corollary}[Centered second-order variance bound]
\label{cor:quantum-centered-bound}
For a fixed estimation operator with arbitrary mean response,
if $I_2^q>(I^q)^2/4$, then
\begin{equation}
 \Delta M^2\geq B_2^q
 :=\frac{[\mu_M''(\theta)-2\mathbb E_{I^q}]^2}
          {I_2^q-(I^q)^2/4}.
 \label{eq:general-centered-root-bound}
\end{equation}
Here $B_2^q$ denotes the centered second-order variance bound
proposed in this work. Under Eq.~\eqref{eq:second-local-unbiased},
\begin{equation}
 B_2^q=\frac{4(\mathbb E_{I^q})^2}{I_2^q-(I^q)^2/4}
       \geq\frac{4(\mathbb E_{I^q})^2}{I_2^q}.
 \label{eq:centered-root-bound}
\end{equation}

\end{corollary}

Normalization gives
\begin{equation}
 4\operatorname{Tr}\!\left[
 \left(R''+\frac{I^q}{4}R\right)^2\right]
 =I_2^q-\frac{(I^q)^2}{4}\geq0.
 \label{eq:quantum-centered-information}
\end{equation}
Centering removes the normalization component of the curvature,
reducing the denominator while preserving the response numerator.
The general-response form covers locally calibrated readouts
with nonzero $\mu_M''(\theta)$, as used in
Sec.~\ref{sec:simulation41}. If the centered denominator vanishes,
its response numerator also vanishes and no ratio is evaluated.

The theorem and corollary are proved together in
Appendix~\ref{app:root-proof}, which also gives the equality
condition and an attainable commuting-state example.

\subsection{Analytical relations and advantage conditions}
\label{sec:root-bh-comparison}
This subsection establishes analytical relations between the proposed second-order information bound and established estimation bounds,
and identifies explicit conditions under which it provides a strictly tighter variance constraint.

We use the working-point notation of Sec.~\ref{sec:root-constraints}
and compare fixed readouts with $\mu_M'(\theta)=1$.
Assume here a full-rank state and independent first- and
second-derivative tests. In the eigenbasis
$\rho=\sum_i\lambda_i|i\rangle\langle i|$, define their real
$2\times2$ SLD information matrix
\begin{equation}
 Q_{jk}=2\sum_{i,l}
 \frac{\operatorname{Re}[\rho^{(j)}_{il}\rho^{(k)}_{li}]}
      {\lambda_i+\lambda_l},\qquad j,k=1,2,
 \label{eq:bh-second-diagonal}
\end{equation}
where $\rho^{(1)}=\rho'$ and $\rho^{(2)}=\rho''$.
Thus $Q_{11}=F_{\mathrm{SLD}}$ and $Q$ is positive definite;
degenerate cases are treated in Appendix~\ref{app:bh-comparison}.
We denote the QCRB and second-order QBh variance bounds by
$B_{\mathrm{QCRB}}$ and $B_{\mathrm{QBh}(2)}$, respectively.
Their response-matched expressions are
\begin{equation}
 \begin{aligned}
 B_{\mathrm{QCRB}}&=\frac1{Q_{11}},\\
 B_{\mathrm{QBh}(2)}&=\frac1{Q_{11}}
 +\frac{[\mu_M''(\theta)-Q_{12}/Q_{11}]^2}
             {Q_{22}-Q_{12}^2/Q_{11}}.
 \end{aligned}
 \label{eq:root-bh-response-comparison}
\end{equation}
The standard second-order locally unbiased choice has
$\mu_M''(\theta)=0$; general calibrated readouts retain
their actual second response in both bounds \cite{Gessner2023}.
The derivative-only Abel branch satisfies
$B_{\mathrm{QAB}(0,s)}=B_{\mathrm{QBh}(s)}$ and
$ B_{\mathrm{QBh}(1)}=B_{\mathrm{QCRB}}$.

At first order $F_{\mathrm{SLD}}\leq I^q\leq2F_{\mathrm{SLD}}$,
with $I^q=F_{\mathrm{SLD}}$ precisely when $[\rho,\rho']=0$.
At second order,
\begin{equation}
 I_2^q=4\sum_{i,j}
 \frac{|\rho''_{ij}-2[(R')^2]_{ij}|^2}
      {(\sqrt{\lambda_i}+\sqrt{\lambda_j})^2}.
 \label{eq:quantum-root-second-spectral}
\end{equation}
Compared with $Q_{22}$, square-root curvature contains the
nonlinear term $2(R')^2$ and different spectral weights.
The classical counterpart is established by the exact mixed- and
fourth-score-moment identity in Appendix~\ref{app:classical-proof}.

This difference yields a constructive model condition. For a
prescribed feasible local calibration, there exists a fixed
Hermitian readout with $B_2^q>B_{\mathrm{QBh}(2)}$ if and only if
\begin{equation}
 (R')^2\notin
 \operatorname{span}_{\mathbb R}\{\rho,\rho',\rho''\}.
 \label{eq:curvature-independent-model}
\end{equation}
Here the real span consists of linear combinations with real coefficients.
An additional readout component orthogonal to these three
operators preserves the calibrated mean and its first two
derivatives while tuning the curvature correlation.
Appendix~\ref{app:bh-comparison} constructs this component and
proves the criterion. For a specified readout satisfying
Eq.~\eqref{eq:second-local-unbiased}, the correlation threshold
can be solved explicitly:
\begin{equation}
 |\mathbb E_{I^q}|>
 \frac12\sqrt{
 \frac{I_2^q-(I^q)^2/4}{Q_{11}-Q_{12}^2/Q_{22}}}.
 \label{eq:general-root-bh-advantage}
\end{equation}
This is the necessary and sufficient condition for strict gain
over QBh(2) in that calibrated class, and also gives strict gain
over QCRB. General-response thresholds are derived in the Appendix~\ref{app:bh-comparison}.

For the mixed-qubit phase model and local sample-mean readout in the following phase estimation application example, the parameter region can also be
solved explicitly. Here $p$ is the Bloch-vector length, $m\geq2$
is the number of independent copies, and $\kappa=\mu_M''(0)$
is the local curvature of the calibrated mean response.
At $\theta=0$, strict gain over the response-matched QBh(2)
holds if and only if
\begin{equation}
 \begin{gathered}
 0<p<\sqrt{\frac{\sqrt3}{2}},\\
 \kappa^2>
 \frac{[1+(2m-1)p^2]
       [2+2\sqrt{1-p^2}+(2m-3)p^2]}
      {2(m-1)p^2(1-p^2)
       [1-2p^2+2\sqrt{1-p^2}]}.
 \end{gathered}
 \label{eq:explicit-finite-copy-region}
\end{equation}
The advantage is therefore controlled by the state mixing,
copy number and curvature of the calibrated readout.
For $m=2$ and $p=0.45$, the threshold is
$|\kappa|>2.887$; the choice $\kappa=6$ in
Sec.~\ref{sec:simulation41} meets this condition.
The model-specific calculation and parameter interval are given
in Appendix~\ref{app:simulation-evaluation}.

\section{Quantum phase estimation}
\label{sec:application}
We illustrate the proposed variance constraint through phase estimation. We compare the proposed bound with a range of established benchmarks
\cite{Helstrom1969,Braunstein1994,Chapman1951,Abel1993,Gessner2023}.

\subsection{Numerical analysis of bounds for qubit phase estimation}
\label{sec:simulation41}

For the numerical comparison, we consider a standard qubit phase-estimation model:
\begin{equation}
 \begin{gathered}
 \rho_\theta=\frac12\left[\mathbb I+
 p(\cos\theta\,\sigma_z+\sin\theta\,\sigma_x)\right],
 \\
 0<p<1,\qquad\theta_0=0.
 \end{gathered}
 \label{eq:simulation-state}
\end{equation}
Here $p$ is the known Bloch-vector length. For $m$ independent copies,
measure the same fixed observable on each qubit and average the outcomes:
\begin{equation}
 \begin{gathered}
 M_\kappa=\frac{\sigma_x-\kappa\sigma_z}{p}+\kappa\mathbb I,
 \\
 \overline M_{m,\kappa}=\frac1m\sum_{j=1}^mM_\kappa^{(j)}.
 \end{gathered}
 \label{eq:simulation-readout}
\end{equation}
The design parameter $\kappa$ specifies a binary projective measurement
along $(1,0,-\kappa)/\sqrt{1+\kappa^2}$, with reported values
$\kappa\pm\sqrt{1+\kappa^2}/p$. Taking its expectation gives
\begin{equation}
 \begin{gathered}
 \mu_m(\theta)=\sin\theta+\kappa(1-\cos\theta),
 \\
 (\mu_m(0),\mu_m'(0),\mu_m''(0))=(0,1,\kappa).
 \end{gathered}
 \label{eq:simulation-response}
\end{equation}
This readout is locally calibrated with unit first derivative.
We evaluate the second-order variance bound using the general-response form of
Corollary~\ref{cor:quantum-centered-bound}; the reference bounds use
the same mean derivatives or the exact difference $\mu_m(\lambda)-\mu_m(0)$.
Independent averaging gives the actual variance
\begin{equation}
 V_m=\frac{1+\kappa^2(1-p^2)}{mp^2}.
 \label{eq:simulation-actual-variance}
\end{equation}

We choose $p=0.45$, $\kappa=6$, and $m=1,\ldots,10$ to show finit-copy advantage.
At $\kappa=6$, the common strict-advantage interval over QBh(2)
for $m=2,\ldots,10$ has endpoints approximately $0.143$ and $0.774$;
the selected state lies inside this common interval.
The closed-form evaluations and parameter conditions are given in
Appendix~\ref{app:simulation-evaluation}. QHCRB and QAB$(1,1)$ are
optimized over $\lambda\in[-\pi,\pi]$ using the complete mean response.

Figure~\ref{fig:simulation-all-bounds} shows that throughout
multi-copy cases, the proposed bound is higher than displayed reference bounds and closer to the actual sample-mean variance.
Especially, relative to QBh(2), its improvement is approximately $14.8\%$ at
$m=2$ and $75.2\%$ at $m=10$.
This realizes the analytical advantage criteria with a fixed local measurement: the additional curvature correlation supplies a stronger finite-copy variance constraint for an explicitly specified readout.

\begin{figure}[!t]

 \centering
 \includegraphics[width=\linewidth]{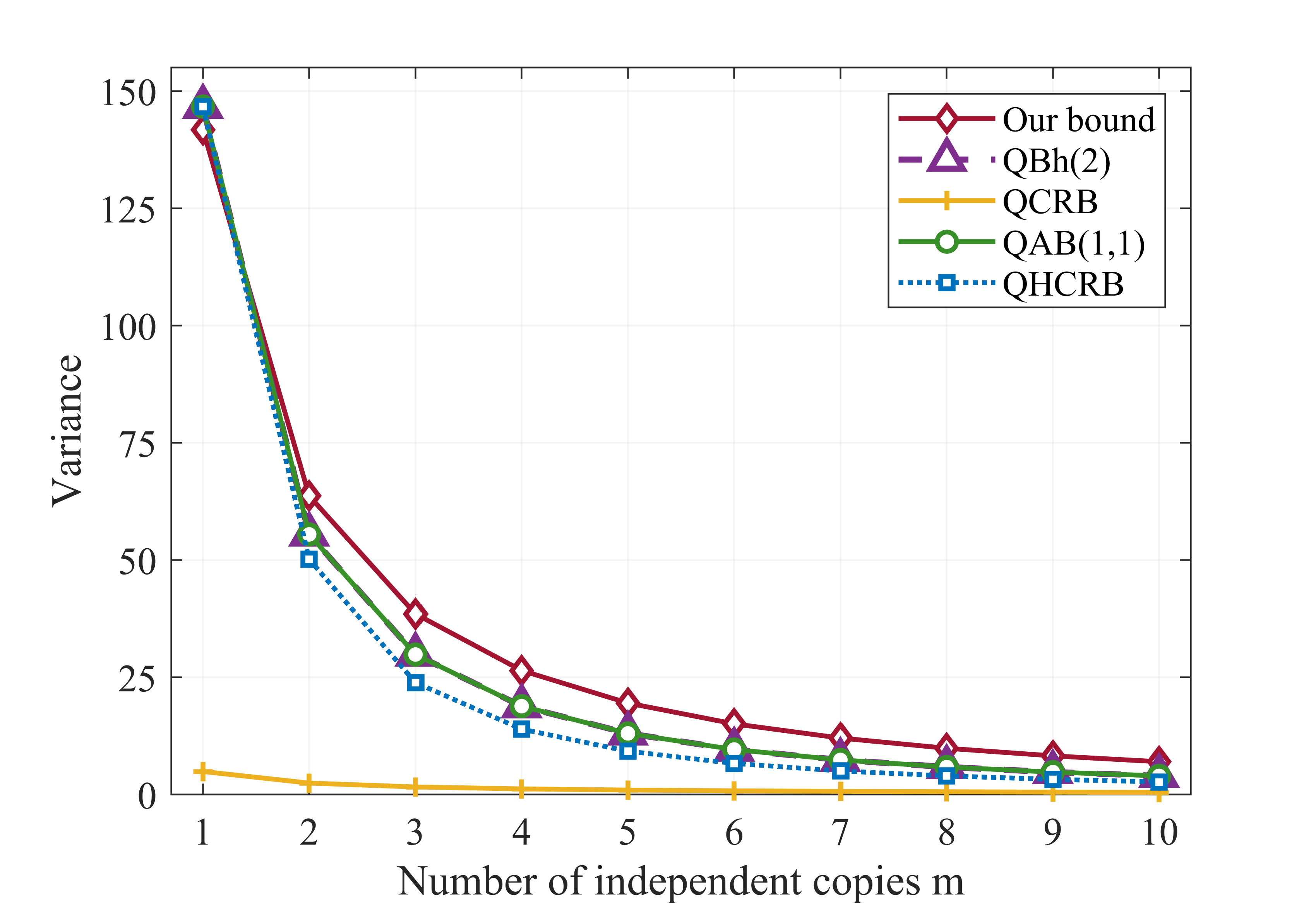}
 \caption{Variance bounds for independent
 qubits at $p=0.45$, $\kappa=6$, $\theta_0=0$, and $m=1,\ldots,10$.
 The centered square-root bound is compared with QCRB, QBh(2),
 QAB$(1,1)$, and QHCRB for the same fixed readout.
 QAB$(1,1)$ and QBh(2) coincide numerically for these parameters.
 Finite-displacement bounds are optimized over $[-\pi,\pi]$.
}
 \label{fig:simulation-all-bounds}
\end{figure}

\subsection{ Experimental validation for single-qubit case}
\subsubsection{Experimental setup}

\begin{figure*}[htbp]
\begin{center}
\includegraphics[width=\textwidth]{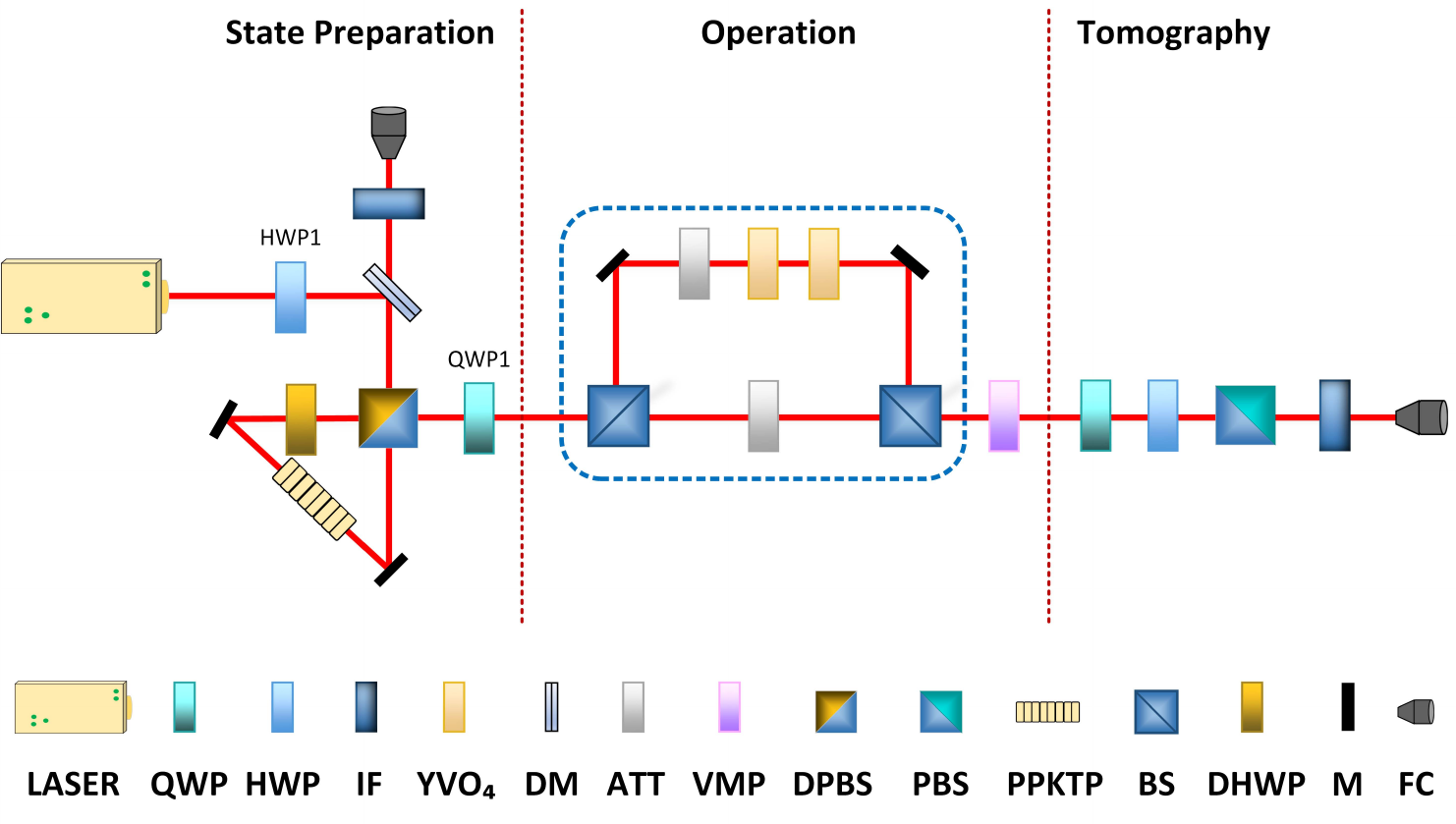}
\caption{Experimental setup, which includes three modules: (a) state preparation module, (b) white-noise insertion and phase-shift control module, (c) quantum state tomography module. In module (a), a pair of maximally polarization-entangled states is generated by a spontaneous parametric down-conversion process at the dual-wavelength polarization beam splitter (DPBS). A quarter-wave plate (QWP1) is then used to prepare the initial single-photon state $|R\rangle = \frac{1}{\sqrt{2}}(|H\rangle + i|V\rangle)$. To produce the single-photon mixed state, the attenuators (ATTs) are used to regulate the mixed weight $p$ of the states. The density matrix is constructed via the quantum state tomography module (c). Key optical elements include half-wave plate (HWP), interference filter (IF), yttrium orthovanadate (YVO$_4$), dichromatic mirror (DM), full-wave liquid-crystal variable wave plate (VWP), polarizing beam splitter (PBS), beam splitter (BS), mirror (M),  dual-wavelength half-wave plate (DHWP), and fiber coupler (FC).}
\label{setupz}
\end{center}
\end{figure*}

The experiment was conducted using the photonic platform shown in Figure \ref{setupz}. The experimental setup consists of three primary components: (1) photon source, (2) white-noise insertion and phase-shift control, and (3) quantum tomography. First, the state-preparation module generates an initial pure state $|\Phi\rangle$. In the white-noise module (dashed blue box), the reflected path dephases one photon into a fully mixed state $\frac{1}{2}\mathbb{I}$. The attenuators (ATTs) are employed to adjust the weight $p$ of the mixed state. The controllable rotation $U(\theta)$ is then introduced using a full-wave liquid-crystal variable wave plate (VWP). In the last part, photon-product measurements are performed by setting the angles of the optical axes of the half-wave plate (HWP) and the quarter-wave plate (QWP).

In the state preparation module shown in Figure \ref{setupz}, a type-II phase-matched periodically poled $\rm{KTiOPO_4}$ (PPKTP) crystal with dimensions $15 \times 2 \times 1$ mm$^3$ is pumped with a 36 mW diode laser beam to generate photon pairs with a central wavelength of 810 nm via the process of spontaneous parametric down-conversion (SPDC). HWP1 in front of the PPKTP crystal is used to control the generated photon pairs, which are encoded in the polarization degree of freedom as $\frac{1}{\sqrt{2}}(|HV\rangle + |VH\rangle)$. One photon acts as a trigger, and the other serves as a signal photon which is further transformed into the state $|R\rangle = \frac{1}{\sqrt{2}}(|H\rangle + i|V\rangle)$ by QWP1.

We experimentally generate a set of noisy states
\begin{eqnarray}
\rho(p) = p |R\rangle\langle R| + \frac{1-p}{2}\mathbb{I}
\end{eqnarray}
where $\mathbb{I}$ denotes the identity operator of rank $2$. As shown in Figure \ref{setupz}, two 50/50 beam splitters (BS) are inserted into the signal branch of the optical setup. In the transmitted path, the photon is prepared in the state $|R\rangle$. In the reflected path, two 2.6 mm yttrium orthovanadate (YVO$_4$) crystals dephase the single-photon state into the maximally mixed state $\frac{\mathbb{I}}{2}$. The ratio of these two components is controlled by two attenuators (ATTs) at the output port of the second BS. After the insertion of white noise, a controllable phase shift $\theta$ is introduced using a full-wave liquid crystal variable wave plate (VWP), which results in the final state $\rho(p,\theta) = U(\theta) \rho(p) U^\dagger(\theta)$, where the rotation is parameterized with a unitary matrix
$ U(\theta) = e^{-i\mathbf{n} \cdot \boldsymbol{\sigma}\theta/2}$. The VWP enables rapid phase adjustments, with $\theta$ switching within approximately 2 ms, allowing for real-time adaptive feedback. This photon source achieves high brightness (0.34 MHz) and a collection efficiency of $60\%$.

Tomographic measurements are conducted on single signal photons, employing a post-selection strategy based on 100 sets of photon coincidences. The measurement results are recorded in coincidence with the trigger photons. Utilizing a 3 nm interference filter, the photon source generates up to 15,000 coincidence counts per second. The conditional probability $P(a|x)$ of the outcome $a \in \{-1,+1\}$ is calculated as
\begin{eqnarray}
p(a|x) = \frac{N_{x}^{a}}{N_{x}^{-1} + N_{x}^{+1}}
\end{eqnarray}
where $N_{x}^{a}$ denotes the photon coincidence counts. The density matrix is reconstructed using maximum-likelihood estimation \cite{ap2006}.

\subsubsection{Experiment results}
We first constructed a set of mixed states $\rho(p)$ by setting different mixture probabilities $\{p,1-p\}$ for photons from two paths in the second module. The prepared state of a single photon can be represented as $\rho_0=\frac{1}{2}(\mathbb{I}+\mathbf{r}_0\cdot \boldsymbol{\sigma})$. In this example, the initial Bloch vector is chosen as $\mathbf{r}=(0,r_0,0)$ with $r_0\in[0,1]$ and $\mathbf{n}$ is the rotation axis. After the preparation stage, we implemented a phase rotation with $\theta=\pi$ along the Pauli $z$-axis, and obtained the output state as
$\rho_\theta=\frac{1}{2}(\mathbb{I}+\mathbf{r}_\theta\cdot \boldsymbol{\sigma})$
with $\mathbf{r}_\theta =(-r_0\sin\theta, r_0\cos\theta,0)$.

To characterize the phase-encoded states, we performed the Pauli
measurements $\sigma_x$, $\sigma_y$, and $\sigma_z$ on the prepared
photons. The measured expectation values give the Bloch vector
$\mathbf{r}_\theta=(\langle\sigma_x\rangle,\langle\sigma_y\rangle,
\langle\sigma_z\rangle)$ and allow the density matrices of the
encoded states to be reconstructed. These reconstructions provide
the basis for evaluating the information quantities and the
reference expressions associated with QCRB \cite{Helstrom1969},
QHCRB
\cite{Barankin1949,Hammersley1950,Chapman1951,Bhattacharyya1946},
and QAB \cite{Abel1993,Gessner2023}.

We evaluated the second-order information according to
Eq.~\eqref{2nd-Qfisher}, using the reconstructed states and the
specified phase-encoding model. Figure~\ref{fig4} displays evaluations
of the uncentered second-order expression with the pointwise and
averaged coefficient choices described in its caption.

\begin{figure*}[t]
    \centering
    \includegraphics[width=\textwidth]{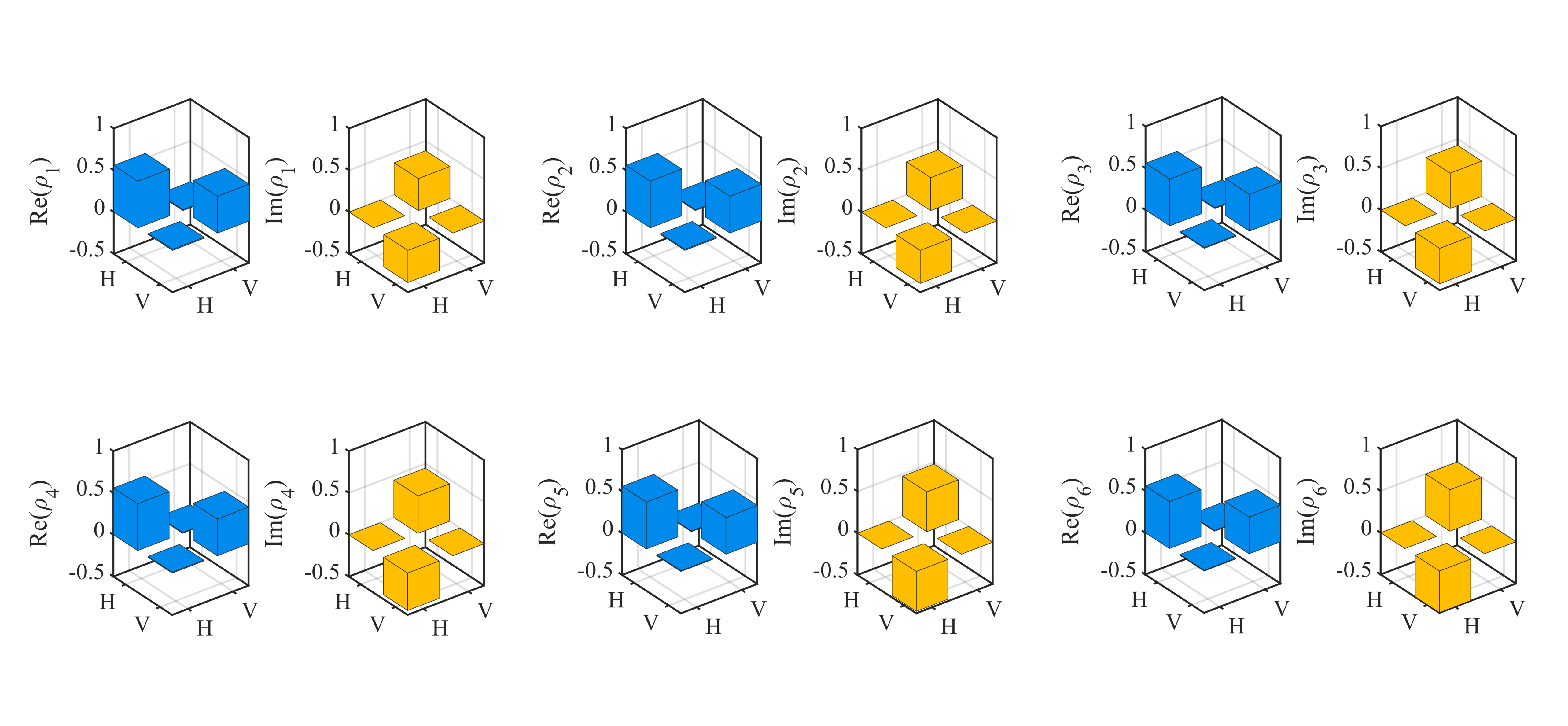}
    \caption{The tomographic results for all the experimental photon states. }
    \label{tomo}
\end{figure*}

\begin{table}[t]
    \caption{Fidelity of experimental single qubit states.}
    \centering
    \begin{tabular}{cc}
    \hline\hline
    $r_0$ & Fidelity\\
    \hline
    0.7500 & $0.9953\pm0.0024$\\
    0.8000 & $0.9956\pm0.0038$\\
    0.8500 & $0.9954\pm0.0030$\\
    0.9000 & $0.9953\pm0.0049$\\
    0.9500 & $0.9949\pm0.0032$\\
    1.0000 & $0.9955\pm0.0036$\\
    \hline\hline
    \end{tabular}
    \label{fid}
\end{table}

\begin{figure}[!t]
    \centering    \includegraphics[width=\linewidth]{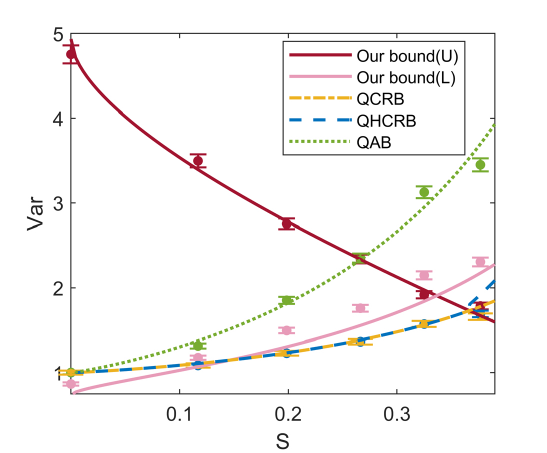}
   \caption{
Comparison of reference-bound expressions and second-order
expression evaluations as a function of the initial Bloch-vector
length $r_0$. Lines denote ideal-model calculations, and dots
denote evaluations from the reconstructed experimental states.
The comparison includes QCRB (yellow dashed line)
\cite{Helstrom1969}, QHCRB (blue dashed line)
\cite{Barankin1949,Hammersley1950,Chapman1951,Bhattacharyya1946},
and QAB (green dotted line) \cite{Abel1993,Gessner2023},
together with the second-order expression evaluations
(red and pink lines).
The red curve uses the pointwise coefficient values obtained
with the Pauli $x$ readout. The pink curve uses their average
across the data points and is included as an illustrative
averaged-coefficient comparison.
The initial Bloch vector is $\mathbf{r}_0=(0,r_0,0)$,
and the applied rotation is $U(\pi)$ about the Pauli $z$ axis.
}
   \label{fig4}
\end{figure}

The quantum state tomography is shown in Figure \ref{tomo}, and fidelity of each quantum state for computing the hierarchical frequentist bounds is shown in Table \ref{fid}. All the fidelities of prepared initial states exceed $99\%$ according to the formula \cite{HHH}:
\begin{eqnarray}
F(\rho_{\text{exp}},\rho_{\text{ideal}})= \mathrm{Tr}\left( \sqrt{\sqrt{\rho_{\text{ideal}}} \rho_{\text{exp}} \sqrt{\rho_{\text{ideal}}}} \right)^2
\end{eqnarray}

Figure~\ref{fig4} compares the evaluated second-order expressions with established reference bounds. In part of the investigated parameter range, the second-order curves take larger values than the displayed reference curves, illustrating their distinct response to state mixing in the single-photon setting.
Whereas the QCRB is governed by first-order Fisher information, the present construction incorporates second-order information through the associated score function. This structure captures additional features of the phase-encoded statistical model and provides a complementary perspective on estimation uncertainty alongside established hierarchical bounds
\cite{Helstrom1969,Barankin1949,Hammersley1950,Chapman1951,Bhattacharyya1946,Abel1993,Gessner2023}.
The reported experimental error bars do not exceed $0.2657$. The single-photon implementation demonstrates the experimental accessibility of the information quantities considered here.
Extending this implementation to joint measurements on multiple copies would provide a further setting for investigating the proposed framework.

The data deviate from the theoretical values within an acceptable error range mainly because of optical-path phase variations or light-field distortions induced by temperature fluctuations, air turbulence, and mechanical vibrations. Moreover, platform noise, particularly in the white-noise insertion and phase-shift control modules, together with imperfections in the noise source and fluctuations of phase-control elements (e.g., the liquid-crystal variable wave plate, VWP) can degrade the phase-shift accuracy and introduce additional measurement errors.

\section{Conclusion}
We have introduced a second-order information measure and derived
corresponding variance bounds for classical and quantum parameter
estimation. Defined through second derivatives of the square roots
of probability densities and quantum density operators, this measure
incorporates second-order statistical structure into the analysis
of estimation uncertainty. The equality conditions clarify how
attainability depends on the statistical model and the chosen readout.

Analytical comparisons with the QCRB and QBh(2) establish explicit
conditions under which the proposed second-order information bound
provides a strictly tighter variance constraint.
Numerical comparisons with these bounds, QHCRB, and QAB(1,1)
demonstrate the predicted advantage in finite-copy qubit phase
estimation. The photonic implementation further demonstrates the experimental
applicability of the proposed framework.

The second-order information is generally nonadditive.Its exact composition law characterizes how first- and second-order
information combine for independent copies. Its spectral representation enables explicit evaluation, and its multiparameter extension organizes second-order sensitivity across parameter directions.

These results establish second-order information as a useful tool for characterizing estimation uncertainty beyond first-order Fisher information. Future work may explore higher-order extensions
and the implications of the multiparameter framework for precision bounds.

\paragraph{Author contributions}

X.Z. and Y.H. conducted the experiments. M.X.L. conducted the research. All authors wrote and reviewed the manuscript.

\paragraph{Funding information}
This work was supported by the National Natural Science Foundation of China (Nos.62671521, 12204386, 12075159, and 12171044), Beijing Natural Science Foundation (No.Z190005), Interdisciplinary Research of Southwest Jiaotong University China (No.2682022KJ004), and the Academician Innovation Platform of Hainan Province.

\appendix
\section{Proofs and equality conditions for the proposed second-order variation bounds}
\label{app:bound-proofs}
In this section, we prove the classical and quantum second-order variance bounds and their centered refinements, and derive the corresponding equality conditions.

Throughout these appendices, $\theta$ denotes an arbitrary parameter
value under discussion. All derivatives are evaluated at that value,
and reported estimates and measurements remain fixed under
differentiation. Local calibration conditions are imposed at that
value. We write $\rho=\rho_\theta$ and $R=\sqrt{\rho_\theta}$.
The required derivatives and moments are assumed finite;
differentiation under classical integrals is assumed valid.

\subsection{Classical bound and centered strengthening}
\label{app:classical-proof}
For a normalized density $p_\theta(x)$ with common support, put
$q=\sqrt{p_\theta}$ and
$Y(x)=\hat\theta(x)-\mu(\theta)$, where
$\mu(\theta)=\int\hat\theta(x)p_\theta(x)\,dx$.
The centering value in $Y$ is held fixed when differentiating.
With $\mathbb E_I=\int Y(q')^2dx$ and
$I_2=4\int(q'')^2dx$, twice differentiating $p_\theta=q^2$ gives
\begin{equation}
 \mu''(\theta)-2\mathbb E_I=2\int Yqq''dx.
 \label{eq:app-classical-response}
\end{equation}
Cauchy--Schwarz therefore yields
\begin{equation}
 \begin{aligned}
 {}[\mu''(\theta)-2\mathbb E_I]^2
 &\leq4\left(\int Y^2q^2dx\right)\left(\int(q'')^2dx\right)\\
 &=\operatorname{Var}_\theta(\hat\theta)I_2.
 \end{aligned}
 \label{eq:app-classical-uncentered-proof}
\end{equation}
For second-order local unbiasedness,
$\mu(\theta)=\theta$, $\mu'(\theta)=1$ and $\mu''(\theta)=0$,
this proves the classical bound in Eq.~\eqref{CR1}.

Writing $I=4\int(q')^2dx$, normalization gives
$\int qq'dx=0$ and $\int qq''dx=-I/4$. Consequently,
\begin{equation}
 \begin{gathered}
 \int q\left(q''+\frac I4q\right)dx=0,\\
 4\int\left(q''+\frac I4q\right)^2dx=I_2-\frac{I^2}{4}\geq0.
 \end{gathered}
 \label{eq:app-classical-centered-norm}
\end{equation}
Since $\int Yq^2dx=0$, this centering leaves
Eq.~\eqref{eq:app-classical-response} unchanged. A second application
of Cauchy--Schwarz proves
\begin{equation}
 \begin{gathered}
 \begin{aligned}
 \operatorname{Var}_\theta(\hat\theta)
 &\geq\frac{[\mu''(\theta)-2\mathbb E_I]^2}{I_2-I^2/4}\\
 &\geq\frac{[\mu''(\theta)-2\mathbb E_I]^2}{I_2},
 \end{aligned}\\
 I_2>I^2/4.
 \end{gathered}
 \label{eq:app-classical-general-centered}
\end{equation}
Thus the same argument supplies an analytical strengthening and
covers a general mean response.

\subsection{Quantum theorem and centered corollary}
\label{app:root-proof}
For a fixed Hermitian first-moment operator $M$, use the local
proof abbreviation $X=M-\mu_M(\theta)\mathbb I$.
The centering value is held fixed in this local calculation.
Differentiating $\rho=R^2$ and using normalization gives
\begin{align}
 \begin{aligned}
 \rho'&=R'R+RR',\\
 \rho''&=R''R+RR''+2(R')^2,
 \end{aligned}
 \label{B-2}\\
 \begin{gathered}
 \operatorname{Tr}(RR')=0,\qquad
 \operatorname{Tr}(XR^2)=0,\\
 \operatorname{Tr}(RR'')=-I^q/4.
 \end{gathered}
 \label{eq:appendix-normalization-identities}
\end{align}
The mean response and the existing correlation satisfy
\begin{equation}
 \begin{gathered}
 \mu_M''(\theta)-2\mathbb E_{I^q}
 =2\operatorname{Re}\operatorname{Tr}(RXR''),
 \\
 \mathbb E_{I^q}=\operatorname{Tr}[X(R')^2].
 \end{gathered}
 \label{eq:appendix-response-identity}
\end{equation}
Since $\|XR\|_{\mathrm{HS}}^2=\Delta M^2$, Hilbert--Schmidt
Cauchy--Schwarz gives
\begin{equation}
 [\mu_M''(\theta)-2\mathbb E_{I^q}]^2
 \leq\Delta M^2 I_2^q.
 \label{eq:appendix-original-proof}
\end{equation}
Setting $\mu_M''(\theta)=0$ proves
Theorem~\ref{thm:second-order-qbound} under the stated local
unbiasedness conditions. This proof only requires differentiability
of the square root, without a full-rank assumption.

For the centered corollary, write
$C_R=R''+(I^q/4)R$ as a local abbreviation for the proof
and equality analysis.
Then
\begin{equation}
 \operatorname{Tr}(RC_R)=0,\qquad
 4\|C_R\|_{\mathrm{HS}}^2=I_2^q-\frac{(I^q)^2}{4}.
 \label{eq:app-centered-root-norm}
\end{equation}
Because $\operatorname{Tr}(RXR)=0$, replacing $R''$ by $C_R$
preserves the response in Eq.~\eqref{eq:appendix-response-identity}.
Hence
\begin{equation}
 [\mu_M''(\theta)-2\mathbb E_{I^q}]^2
 \leq\Delta M^2\left[I_2^q-\frac{(I^q)^2}{4}\right],
 \label{eq:app-quantum-centered-proof}
\end{equation}
which proves Corollary~\ref{cor:quantum-centered-bound}.
If the centered denominator vanishes, $C_R=0$ and the associated
response is zero; no ratio is assigned.

For a POVM with reported values $h(x)$, put
$N=\sum_xh(x)^2E_x$. The positive-operator identity
\begin{equation}
 N-M^2=\sum_x[h(x)\mathbb I-M]E_x[h(x)\mathbb I-M]\succeq0
\end{equation}
implies
\begin{equation}
 \operatorname{Var}_\theta(\hat\theta)-\Delta M^2
 =\operatorname{Tr}[\rho(N-M^2)]\geq0.
 \label{eq:app-povm-variance-proof}
\end{equation}
This transfers both bounds to the actual reported-estimate variance.
A projective measurement of $M$ has $N=M^2$.

\subsection{Equality conditions and an attainable case}
For positive centered denominators and nonzero responses, equality
in the classical and quantum centered bounds holds precisely when
\begin{align}
 Yq&=\frac{2[\mu''(\theta)-2\mathbb E_I]}{I_2-I^2/4}
       \left(q''+\frac I4q\right),
 \label{eq:app-classical-equality}\\
 XR&=\frac{2[\mu_M''(\theta)-2\mathbb E_{I^q}]}
             {I_2^q-(I^q)^2/4}\,C_R.
 \label{eq:centered-equality-local}
\end{align}
These are the real proportionality conditions for the two
Cauchy--Schwarz inequalities. For the uncentered bounds, replace
$q''+(I/4)q$ by $q''$ and its denominator by $I_2$, and replace
$C_R$ by $R''$ and its denominator by $I_2^q$, respectively.
For full-rank $\rho$, Hermiticity
of the second expression requires $[R,R'']=0$; the prescribed
mean responses must also hold. Equality for the actual POVM
variance additionally requires $\operatorname{Tr}[\rho(N-M^2)]=0$.

An explicit attainable family is
$\rho_\theta=(\mathbb I+\theta\sigma_z)/2$, $M=\sigma_z$,
with $0<|\theta|<1$. Here $\mu_M(\theta)=\theta$ and
\begin{equation}
 \mathbb E_{I^q}=-\frac{\theta}{2(1-\theta^2)},\qquad
 I_2^q-\frac{(I^q)^2}{4}=\frac{\theta^2}{(1-\theta^2)^3}.
\end{equation}
Direct substitution gives
\begin{equation}
 B_2^q=\frac{4(\mathbb E_{I^q})^2}{I_2^q-(I^q)^2/4}
      =1-\theta^2=\Delta M^2
      =\operatorname{Var}_\theta(\hat\theta).
 \label{eq:app-centered-attainment}
\end{equation}
Thus the strengthened bound has a nonzero attainable instance.

\setcounter{figure}{7}
\subsection{Application: Direction-of-arrival parameter estimation.}
\begin{figure}[!t]
    \centering
\includegraphics[width=\linewidth]{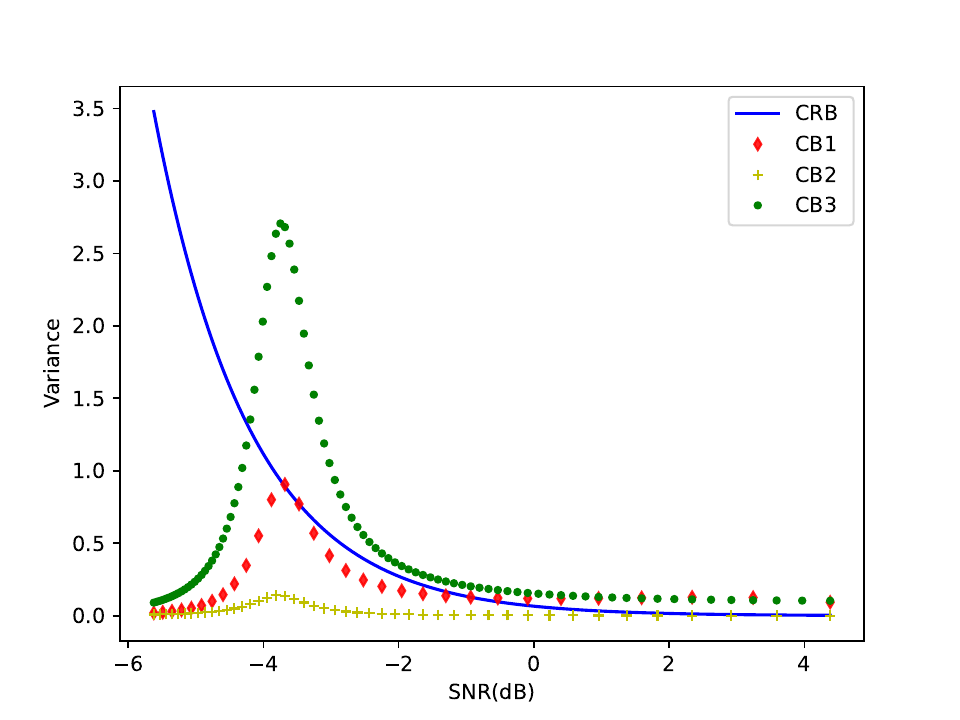}
    \caption{The comparison of the classical lower bounds of the variance of a direction of arrival estimation problem. The blue line denotes the Cramer-Rao bound (CRB) and the others are the proposed classical bounds. the present new classical bound based on higher-order information metric corresponding to the sine estimator (CB1) and cosine estimator (CB2). Point estimation with $(\theta-\|\hat{\theta}\|_2)^2=0.1$ (CB3).}
    \label{Figure2}
\end{figure}

To show explicitly the distance between the propose classical bound and CRB, we
give an example of estimating the arrival direction of signal in radar systems.

Angle of arrival estimation in modern radar systems relies on antenna arrays for spatial data capture and beamforming implementation. The array's aperture significantly impacts angular resolution, while array density enhances main beam suppression and prevents spatial aliasing, favoring large, densely packed arrays despite associated higher costs. Sensor selection involves choosing a subset of sensors from a full set of candidates, with compressed sensing proving effective in reducing data processing demands while sustaining performance, though the full array remains essential for low-complexity analog processing.

Here, we simplify the array model by focusing on parameter estimation for signals received by a single sensor in a one-dimensional, single-parameter estimation scenario. This approach aims to tackle hardware and computational challenges by potentially enabling a more sparse array design, where only specific sensors are activated, optimizing performance and resource utilization. Signal processing issues such as direction finding with narrow-band sensor arrays and resolving overlapping echoes heavily rely on robust parameter estimation. Accurately estimating parameters of multiple superimposed exponential signals in noise is crucial for addressing these challenges. The Cram\'{e}r-Rao bound has traditionally been a go-to for offline sensor selection, providing a measure of estimation performance regardless of the method used. Additionally, we evaluate the second-order FI to show that the information lower bound introduced in the main text could serve as an effective alternative.

Consider a model to estimate parameter $\theta$ in the transfer item as
\begin{eqnarray}
    y(t) = A(\theta)x(t) + e(t) \quad t = 1, 2, \ldots, N,
\end{eqnarray}
where $y(t)$ is the noisy data vector, $x(t)$ is the vector of signal amplitudes, $e(t)$ is additive noise. $A(\theta)$ is transfer matrix defined as $A(\theta)=[a(\omega_1) \ldots a(\omega_n)]$, where $\omega_i$ is the wave frequency, $a(\omega_i)$ is transfer vector, containing the information of target direction of arrival, and $\theta = [\omega_1 \ldots \omega_n]^T$.

The likelihood function in this case is given by:
\begin{equation}
\begin{aligned}
 &L(y(1), \ldots, y(N))\\
 &\quad= \frac{1}{(2\pi)^{mN} (\sigma^2/2)^{mN}}\\
 &\qquad\times\exp\Biggl(-\frac{1}{\sigma^2} \sum_{t=1}^{N}
 [y(t) - Ax(t)]^*\\
 &\hspace{94pt}\times[y(t) - Ax(t)]\Biggr).
\end{aligned}
\end{equation}
Now, we simplify the problem as $m=1$ and $N=1$ to estimate a scalar $\omega$ to implement one parameter estimation with one dimension signal and a single source.
Let $x$ be the signal amplitude. The transfer item is simplified as $a(\omega)=e^{i\omega}$. By taking the derivative of $\omega$ of its log-likelihood function, it follows that
\begin{eqnarray}
\frac{\partial \ln L}{\partial \omega} = \frac{2}{\sigma^2} \sum_{t=1}^{N} \Re\left\{ x^*(t) \frac{\partial A^*}{\partial \omega} e(t) \right\},
\end{eqnarray}
where $\Re$ denotes the real part. Then the FI is given by
\begin{eqnarray}
    I=\mathbb{E}\left[\frac{\partial \ln L}{\partial \omega} \right]^2.
\end{eqnarray}
Moreover, the second-order FI we proposed in the main text can be calculated as
\begin{eqnarray}
    I_2(X,\omega)=\mathbb{E}\left[ \frac{\partial^2 \ln L}{\partial \omega^2}  +  \frac{1}{2}(\frac{\partial \ln L}{\partial \omega} )^2\right]^2.
\end{eqnarray}
where $\ln L =const - m\ln\sigma-\frac{1}{\sigma}[y-a(\omega)x]^*[y-a(\omega)x]$ is the log-likelihood function of the noisy data. Here, $a^*$ denotes the conjugate of $a$. In the simulation, we choose the settings as $\omega = \frac{\pi}{3}$, signal amplitude as $\frac{\pi}{6}$, $\sigma^2 \in [0.1,1]$ and $x \in [-0.2, 0.2]$. Expanding the integral interval will lead to an increase of the bounds. This means the present second-order FI is more efficient for parameter estimation when the samples are close to some fixed values.

Overall, we simplify the common array model by focusing on parameter estimation for the signal received by a single sensor within the one-dimensional, single-parameter estimation problem. We compare the standard Cram{\'e}r-Rao bound (CRB) and the proposed new classical bound for sine estimator(CB1) and cosine estimator(CB2). For the point estimation scenario, a small deviation of the estimator (i.e., $(\theta-\|\hat{\theta}\|_2)^2=0.1$ in Fig.\ref{Figure2} (CB3)  demonstrates an improved performance beyond the CRB. In general, increasing the signal-to-noise ratio (SNR) will reduce the present bound to the CRB.

\section{Comparison with Bhattacharyya bounds}
\label{app:bh-comparison}

\subsection{Analytical comparison of the QCRB and QBh(2)}
At the parameter value under discussion, assume full-rank $\rho$
and compare fixed readouts with $\mu_M'(\theta)=1$. Define the SLD division map and
second-order derivative information matrix by
\begin{equation}
 \rho\Omega_\rho(A)+\Omega_\rho(A)\rho=2A,
 \label{eq:app-sld-division-map}
\end{equation}
\begin{equation}
 \begin{aligned}
 Q_{jk}&=\operatorname{Re}\operatorname{Tr}
 [\rho^{(j)}\Omega_\rho(\rho^{(k)})]
 \\
 &=2\sum_{i,l}\frac{\operatorname{Re}
 [\rho^{(j)}_{il}\rho^{(k)}_{li}]}{\lambda_i+\lambda_l},
 \qquad j,k=1,2.
 \end{aligned}
 \label{eq:quantum-bh-spectral-matrix}
\end{equation}
Here $\lambda_i>0$ are the eigenvalues of $\rho$,
$\rho^{(1)}=\rho'$, $\rho^{(2)}=\rho''$, and
$Q_{11}=F_{\mathrm{SLD}}$.
In the real inner product
$\langle A,B\rangle_\rho=\operatorname{Re}\operatorname{Tr}(\rho AB)$,
the zero-mean tests $\Omega_\rho(\rho^{(j)})$ have Gram matrix $Q$
and correlations $\mu_M^{(j)}(\theta)$ with
$M-\mu_M(\theta)\mathbb I$. Orthogonal projection therefore gives
\begin{equation}
 \begin{aligned}
 B_{\mathrm{QBh}(2)}&=\boldsymbol v^{\mathsf T}Q^{-1}\boldsymbol v,\\
 \boldsymbol v&=(1,\mu_M''(\theta))^{\mathsf T},
 \end{aligned}
 \label{eq:app-bh-projection-bound}
\end{equation}
for $Q\succ0$ and $\mu_M'(\theta)=1$ \cite{Gessner2023}.
This retains the same actual response as the proposed bound.
Inverting the $2\times2$ matrix yields the explicit gap
\begin{equation}
 \begin{aligned}
 B_{\mathrm{QBh}(2)}-B_{\mathrm{QCRB}}
 &=\frac{[\mu_M''(\theta)-Q_{12}/Q_{11}]^2}
 {Q_{22}-Q_{12}^2/Q_{11}},
 \\
 B_{\mathrm{QCRB}}&=B_{\mathrm{QBh}(1)}=\frac1{Q_{11}}.
 \end{aligned}
 \label{eq:app-bh-schur-form}
\end{equation}
The conventional second-order locally unbiased case sets
$\mu_M''(\theta)=0$; the gap then vanishes precisely when $Q_{12}=0$.
For other responses it vanishes precisely when
$\mu_M''(\theta)=Q_{12}/Q_{11}$.
If $Q$ is singular, $Q^{-1}$ is replaced by its Moore--Penrose
inverse $Q^+$: any physically attainable response belongs to the
column space of $Q$, since a null Gram direction has zero correlation
with every readout.

The first-order comparison $F_{\mathrm{SLD}}\le I^q\le2F_{\mathrm{SLD}}$
and the second-order spectral formula
\eqref{eq:app-root-second-spectral} distinguish the information itself
from the derivative Gram matrix. In particular, $I_2^q$ includes
$\rho''-2(R')^2$ and the weights
$(\sqrt{\lambda_i}+\sqrt{\lambda_l})^{-2}$; $Q_{22}$ uses
$\rho''$ and $(\lambda_i+\lambda_l)^{-1}$.
The readout correlation determines how this distinct information
structure strengthens the variance constraint.

For the classical counterpart, define $s_1=p'/p$, $s_2=p''/p$
and $\mathsf B_{22}=\mathbb E[s_2^2]$, for $p=p_\theta>0$ on the
common support. Direct differentiation gives the exact relation
\begin{equation}
 \begin{aligned}
 I_2&=\mathsf B_{22}-\mathbb E[s_1^2s_2]
                  +\tfrac14\mathbb E[s_1^4],\\
 I_2-I^2/4&=\operatorname{Var}(s_2-\tfrac12s_1^2).
 \end{aligned}
 \label{eq:classical-root-bh-identity}
\end{equation}
Here the expectation uses $p_\theta$. The mixed and fourth-score
moments resolve the second-order structure beyond
$\mathsf B_{22}$.

\subsection{Conditions for an advantage over QBh(2)}
For $Q\succ0$ and $I_2^q-(I^q)^2/4>0$, solving
$B_2^q>B_{\mathrm{QBh}(2)}$ gives the necessary and sufficient condition
\begin{widetext}
\begin{equation}
 \left|\mathbb E_{I^q}-\frac{\mu_M''(\theta)}2\right|
 >\frac12\sqrt{\frac{[I_2^q-(I^q)^2/4]
 [Q_{22}-2\mu_M''(\theta)Q_{12}+\mu_M''(\theta)^2Q_{11}]}
 {Q_{11}Q_{22}-Q_{12}^2}}.
 \label{eq:app-explicit-correlation-threshold}
\end{equation}
\end{widetext}
Exceeding this threshold guarantees strict gain over both
QBh(2) and QCRB. For second-order local unbiasedness it reduces to
\begin{equation}
 |\mathbb E_{I^q}|>
 \frac12\sqrt{\frac{I_2^q-(I^q)^2/4}{Q_{11}-Q_{12}^2/Q_{22}}}.
 \label{eq:app-local-unbiased-threshold}
\end{equation}
The weaker QCRB threshold is obtained by replacing the denominator
in the last expression by $Q_{11}$.

A model permits such a readout, with any prescribed feasible local
mean and first two responses, if and only if
\begin{equation}
 (R')^2\notin\operatorname{span}_{\mathbb R}
 \{\rho,\rho',\rho''\}.
 \label{eq:app-advantage-structure}
\end{equation}
Here the span consists of real linear combinations of its
Hermitian generators. To prove sufficiency, let $A$ be the
Hilbert--Schmidt projection of $(R')^2$ onto its orthogonal complement.
Then
\begin{equation}
 \begin{aligned}
 \operatorname{Tr}(\rho A)&=\operatorname{Tr}(\rho'A)
 =\operatorname{Tr}(\rho''A)=0,\\
 \operatorname{Tr}[A(R')^2]&=\operatorname{Tr}(A^2)>0.
 \end{aligned}
 \label{eq:app-advantage-readout-component}
\end{equation}
The fixed Hermitian readout $M_t=M+tA$ preserves those three
calibration values, while
$\mathbb E_{I^q}(M_t)=\mathbb E_{I^q}(M)+t\operatorname{Tr}(A^2)$.
A finite real design amplitude $t$ can thus cross
\eqref{eq:app-explicit-correlation-threshold}.

For necessity, introduce only the local test
$G=\rho''-2(R')^2+(I^q/2)\rho$, whose trace is zero.
Failure of \eqref{eq:app-advantage-structure} is equivalent to
$G$ lying in the real span of $\rho'$ and $\rho''$.
Its SLD norm obeys
\begin{equation}
 \begin{aligned}
 &\operatorname{Tr}[G\Omega_\rho(G)]\\
 &\quad=2\sum_{i,l}\frac{(\sqrt{\lambda_i}+\sqrt{\lambda_l})^2}
 {\lambda_i+\lambda_l}
 \left|\left[R''+\frac{I^q}{4}R\right]_{il}\right|^2\\
 &\quad\le I_2^q-\frac{(I^q)^2}{4}.
 \end{aligned}
 \label{eq:root-general-test-framework}
\end{equation}
If $G$ is in that derivative span, Gram-space Cauchy--Schwarz bounds
its squared response divided by its SLD norm by
$\boldsymbol v^{\mathsf T}Q^{-1}\boldsymbol v$.
Its response equals $\mu_M''(\theta)-2\mathbb E_{I^q}$, so the
larger denominator in $B_2^q$ precludes strict gain.
The zero-norm case has zero response.
This proves the criterion; it establishes existence of a calibrated
Hermitian readout, while \eqref{eq:app-explicit-correlation-threshold}
tests the specified readout used in an experiment.
\subsection{Relations between quantum Abel and Bhattacharyya bounds}

The quantum Abel family connects finite-displacement and derivative
bounds. With the same derivative tests and response constraints,
its derivative-only branch coincides exactly with the quantum
Bhattacharyya hierarchy \cite{Gessner2023}:
\begin{equation}
 B_{\mathrm{QAB}(0,s)}=B_{\mathrm{QBh}(s)}.
 \label{eq:app-abel-bh-branch}
\end{equation}

For QAB$(1,1)$, the tests are $\rho'$ and
$\rho_{\theta+\lambda}-\rho_\theta$.
For an admissible displacement $\lambda\ne0$, we retain $\rho'$
and replace the second test by subtracting $\lambda\rho'$
and multiplying by $2/\lambda^2$.
Twice differentiability gives
\begin{equation}
 \frac{2[\rho_{\theta+\lambda}-\rho_\theta-\lambda\rho']}
      {\lambda^2}
 \xrightarrow{\lambda\to0}\rho''.
 \label{eq:qab-derivative-limit}
\end{equation}
Applying the same transformation to the responses
$\mu_M'(\theta)$ and
$\mu_M(\theta+\lambda)-\mu_M(\theta)$ makes the second response
converge to $\mu_M''(\theta)$.
This simultaneous invertible transformation of tests and
responses leaves the quadratic bound unchanged.

For full-rank $\rho$ and $Q\succ0$, the transformed Gram matrix
converges to $Q$, and continuity of its inverse yields
\begin{equation}
 \lim_{\lambda\to0}B_{\mathrm{QAB}(1,1)}(\lambda)
 =B_{\mathrm{QBh}(2)}.
 \label{eq:qab-bh-second-limit}
\end{equation}
Thus, when arbitrarily small nonzero displacements are admissible,
the supremum over displacements is at least
$B_{\mathrm{QBh}(2)}$.

\section{Evaluation and properties of the second-order information}
\label{app:information-properties}

The identities
\begin{equation}
 \begin{aligned}
 RR'+R'R&=\rho',\\
 RR''+R''R&=\rho''-2(R')^2
 \end{aligned}
 \label{eq:app-information-sylvester}
\end{equation}
give a direct evaluation of $I_2^q=4\operatorname{Tr}[(R'')^2]$.
For full-rank $\rho$, their Lyapunov integral solutions are
\begin{align}
 R'&=\int_0^\infty e^{-tR}\rho'e^{-tR}\,dt,\nonumber\\
 R''&=\int_0^\infty
 e^{-tR}[\rho''-2(R')^2]e^{-tR}\,dt.
 \label{eq:app-information-lyapunov}
\end{align}
In the eigenbasis of $\rho$, chosen after differentiation, this gives
\begin{equation}
 \begin{aligned}
 R'_{ij}&=\frac{\rho'_{ij}}{\sqrt{\lambda_i}+\sqrt{\lambda_j}},\\
 R''_{ij}&=\frac{\rho''_{ij}-2[(R')^2]_{ij}}
                   {\sqrt{\lambda_i}+\sqrt{\lambda_j}},
 \end{aligned}
 \label{eq:quantum-root-spectral-derivatives}
\end{equation}
and hence
\begin{equation}
 I_2^q=4\sum_{i,j}
 \frac{|\rho''_{ij}-2[(R')^2]_{ij}|^2}
      {(\sqrt{\lambda_i}+\sqrt{\lambda_j})^2}.
 \label{eq:app-root-second-spectral}
\end{equation}

The same representation gives
\begin{equation}
 \begin{aligned}
 I^q&=4\sum_{i,j}\frac{|\rho'_{ij}|^2}
                {(\sqrt{\lambda_i}+\sqrt{\lambda_j})^2},\\
 F_{\mathrm{SLD}}&=2\sum_{i,j}\frac{|\rho'_{ij}|^2}{\lambda_i+\lambda_j}.
 \end{aligned}
 \label{eq:app-first-information-spectral}
\end{equation}

Since $\lambda_i+\lambda_j\leq(\sqrt{\lambda_i}+\sqrt{\lambda_j})^2
\leq2(\lambda_i+\lambda_j)$, termwise comparison proves
\begin{equation}
 F_{\mathrm{SLD}}\leq I^q\leq2F_{\mathrm{SLD}},
 \label{eq:root-sld-first-comparison}
\end{equation}
with $I^q=F_{\mathrm{SLD}}$ precisely when $[\rho,\rho']=0$.

The second-order expression retains the nonlinear correction
generated by the first derivative. The integral solution requires full rank;
the definition of $I_2^q$ remains applicable whenever $R''$ exists
and has finite Hilbert--Schmidt norm.

The second-order information has the following useful structural
properties. Normalization gives
$\operatorname{Tr}(RR')=0$ and
$\operatorname{Tr}(RR'')=-I^q/4$, so
\begin{equation}
 I_2^q-\frac{(I^q)^2}{4}
 =4\left\|R''+\frac{I^q}{4}R\right\|_{\mathrm{HS}}^2\geq0.
 \label{eq:app-information-centered-property}
\end{equation}
For a parameter-independent unitary $U$, both derivatives transform
by conjugation, and the Hilbert--Schmidt norms defining $I^q$ and $I_2^q$
are unchanged. If
$\rho_\theta=\sum_i p_i(\theta)|i\rangle\langle i|$ has a fixed
eigenbasis, then
\begin{equation}
 I_2^q=4\sum_i[(\sqrt{p_i})'']^2=I_2,
 \label{eq:app-information-classical-reduction}
\end{equation}
which recovers the classical second-order information.

For $m$ independent identical copies, differentiation of
$\sqrt{\rho_\theta^{\otimes m}}=R^{\otimes m}$ gives
\begin{equation}
 \begin{aligned}
 I_m^q&=mI^q,\\
 I_{2,m}^q&=mI_2^q+\frac34m(m-1)(I^q)^2.
 \end{aligned}
 \label{eq:app-iid-second-information}
\end{equation}
To verify the second identity, the second derivative contains
$m$ terms with one $R''$ and $\binom m2$ terms with two $R'$,
the latter each carrying a factor two. Their squared norms are
$mI_2^q/4+m(m-1)(I^q/4)^2$ for the sum of the single-site
terms and $2m(m-1)(I^q/4)^2$ for the sum of the two-site terms.
All remaining cross terms vanish because
$\operatorname{Tr}(RR')=0$. Multiplication by four yields
Eq.~\eqref{eq:app-iid-second-information}. This composition law
also determines the centered information for any copy number
without a matrix calculation in the full tensor-product space.

\section{Multiparameter second-order information}
\label{app:multiparameter-information}

In the chosen parametrization, consider a regular density
$p_{\boldsymbol\theta}(x)>0$ on the common support and let
$\ell=\ln p_{\boldsymbol\theta}$. Derivatives are evaluated at
the parameter value under discussion, and differentiation under
the integral is assumed valid. The second-order information
matrix associated with the pure parameter directions is
\begin{align}
 [I_2]_{ij}
 &=4\int
 (\partial_{ii}\sqrt{p_{\boldsymbol\theta}})
 (\partial_{jj}\sqrt{p_{\boldsymbol\theta}})\,dx
 \nonumber\\
 &=\mathbb E\!\left[
 \left(\partial_{ii}\ell+\tfrac12(\partial_i\ell)^2\right)
 \left(\partial_{jj}\ell+\tfrac12(\partial_j\ell)^2\right)
 \right].
 \label{eq:app-multiparameter-second-information}
\end{align}
Here $\partial_{ij}=\partial_{\theta_i}\partial_{\theta_j}$ and
$\partial_i=\partial_{\theta_i}$.
This is a Gram matrix of the pure second-derivative directions.
For every real vector $\boldsymbol c$,
\begin{equation}
 \boldsymbol c^{\mathsf T}I_2\boldsymbol c
 =4\int\left(\sum_i c_i
     \partial_{ii}\sqrt{p_{\boldsymbol\theta}}\right)^2dx\geq0,
 \label{eq:app-multiparameter-psd}
\end{equation}
proving positive semidefiniteness for every such regular model.

Consider the Gaussian density parametrized by its mean and variance,
\begin{equation}
 p_{\mu,\sigma^2}(x)=\frac1{\sqrt{2\pi\sigma^2}}
 \exp\!\left[-\frac{(x-\mu)^2}{2\sigma^2}\right],
 \qquad \sigma^2>0.
\end{equation}
Using
$\mathbb E[(x-\mu)^{2k}]=(2k-1)!!\,\sigma^{2k}$ gives
\begin{equation}
 \begin{aligned}
 I(\mu,\sigma^2)&=
 \begin{pmatrix}
  1/\sigma^2&0\\[1mm]0&1/(2\sigma^4)
 \end{pmatrix},\\[1mm]
 I_2(\mu,\sigma^2)&=
 \begin{pmatrix}
  \dfrac{3}{4\sigma^4}&-\dfrac{3}{8\sigma^6}\\[2mm]
  -\dfrac{3}{8\sigma^6}&\dfrac{15}{16\sigma^8}
 \end{pmatrix}.
 \end{aligned}
 \label{eq:app-gaussian-second-information}
\end{equation}
In particular, $\det I_2=9/(16\sigma^{12})>0$.
The zero off-diagonal element of $I$ expresses orthogonality of
the first-order score directions. The nonzero off-diagonal element
of $I_2$ instead resolves their second-order coupling. Thus the
second-order information describes a parameter structure beyond
the first-order Fisher matrix in this example.

Mixed directions are included by the full second-derivative Gram
array
\begin{equation}
 [I_2]_{ij,kl}=4\int
 (\partial_{ij}\sqrt{p_{\boldsymbol\theta}})
 (\partial_{kl}\sqrt{p_{\boldsymbol\theta}})\,dx.
 \label{eq:app-mixed-second-information}
\end{equation}
Along a straight parameter path
$\boldsymbol\theta(t)=\boldsymbol\theta+t\boldsymbol u$,
the scalar second-order information is
\begin{equation}
 I_2[\boldsymbol u]=\sum_{i,j,k,l}u_i u_j u_k u_l[I_2]_{ij,kl}.
 \label{eq:app-directional-information}
\end{equation}
The preceding matrix is the pure-direction block, whose $(i,j)$
entry is $[I_2]_{ii,jj}$. The quantum extension is
\begin{equation}
 [I_2^q]_{ij,kl}=4\operatorname{Re}\operatorname{Tr}\!\left[
 (\partial_{ij}\sqrt{\rho_{\boldsymbol\theta}})
 (\partial_{kl}\sqrt{\rho_{\boldsymbol\theta}})\right].
 \label{eq:app-quantum-mixed-information}
\end{equation}
Its Gram form likewise proves positive semidefiniteness. These
constructions extend the information quantity; covariance
constraints additionally require specified response conditions.

\section{Analytical evaluation of the simulation bounds}
\label{app:simulation-evaluation}

\subsection{Information and bounds for the qubit phase-estimation example}
For the model and fixed sample-mean readout of
Sec.~\ref{sec:simulation41}, all quantities below are evaluated at
$\theta=0$ with the actual derivative response $(1,\kappa)^{\mathsf T}$.
Writing $a=1-\sqrt{1-p^2}$ simplifies the repeated information formulas:
\begin{equation}
 \begin{aligned}
 I_m^q&=2ma,\\
 I_{2,m}^q&=2ma+3m(m-1)a^2,\\
 \mathbb E_{I^q,m}&=\frac{\kappa a}{2}.
 \end{aligned}
 \label{eq:app-product-information}
\end{equation}
These expressions follow from the tensor-product composition law
and the $1/m$ factor in the readout. The derivative Gram matrix is
diagonal, with
\begin{equation}
 \begin{aligned}
 Q_{11}&=mp^2,\\
 Q_{12}&=0,\\
 Q_{22}&=\frac{mp^2[1+(2m-1)p^2]}{(1-p^2)(1+p^2)}.
 \end{aligned}
 \label{eq:app-product-gram}
\end{equation}
We denote $B_2^q$ evaluated on this readout by $B_{2,m}^q$.
Substitution gives
\begin{equation}
 \begin{aligned}
 B_{2,m}^q&=\frac{\kappa^2(1-p^2)}{m[2a+(2m-3)a^2]},
  \\
  B_{\mathrm{QCRB}}^{(m)}&=\frac1{mp^2},\\
 B_{\mathrm{QBh}(2)}^{(m)}
 &=\frac1{mp^2}+\frac{\kappa^2(1-p^2)(1+p^2)}
  {mp^2[1+(2m-1)p^2]},\\
 \operatorname{Var}_{0}(\hat\theta)
  &=\frac{1+\kappa^2(1-p^2)}{mp^2}.
 \end{aligned}
 \label{eq:simulation-bound-values}
\end{equation}
Solving the comparison yields the parameter region in
Eq.~\eqref{eq:explicit-finite-copy-region}.
For $\kappa=6$, its common interval for $m=2,\ldots,10$ is
$0.143\lesssim p\lesssim0.774$, with endpoints excluded
for strict gain. The selected $p=0.45$ gives an increase over
QBh(2) of approximately $14.8\%$ at $m=2$ and $75.2\%$ at $m=10$.

\subsection{Evaluation of QHCRB and QAB$(1,1)$}
Write $\rho_m=\rho_0^{\otimes m}$ and
$\delta\rho_\lambda=\rho_\lambda^{\otimes m}-\rho_m$.
Using the SLD division map defined in
Eq.~\eqref{eq:app-sld-division-map}, set
$\chi_\lambda=\operatorname{Tr}
[\delta\rho_\lambda\Omega_{\rho_m}(\delta\rho_\lambda)]$.
The derivative and displacement tests have Gram matrix and responses
\begin{equation}
 \begin{pmatrix}
 mp^2&mp^2\sin\lambda\\
 mp^2\sin\lambda&\chi_\lambda
 \end{pmatrix},\qquad
 \begin{pmatrix}1\\\sin\lambda+\kappa(1-\cos\lambda)\end{pmatrix}.
 \label{eq:app-finite-gram}
\end{equation}
Thus the response-matched finite-displacement bounds are
\begin{align}
 B_{\mathrm{QHCR}}^{(m)}(\lambda)
 &=\frac{[\sin\lambda+\kappa(1-\cos\lambda)]^2}{\chi_\lambda},
 \nonumber\\
 B_{\mathrm{QAB}(1,1)}^{(m)}(\lambda)
 &=\frac1{mp^2}+\frac{\kappa^2(1-\cos\lambda)^2}
 {\chi_\lambda-mp^2\sin^2\lambda}.
 \label{eq:app-finite-bounds}
\end{align}
For small displacements the second denominator is evaluated as the
positive SLD norm of the orthogonal residual
$\delta\rho_\lambda-\sin\lambda\,\rho_m'$, avoiding cancellation.
Here $\rho_m'$ is the derivative of the $m$-copy state family
evaluated at $\theta=0$.
The continuous limits are
\begin{equation}
 \begin{aligned}
 \lim_{\lambda\to0}B_{\mathrm{QHCR}}^{(m)}(\lambda)
  &=B_{\mathrm{QCRB}}^{(m)},\\
 \lim_{\lambda\to0}B_{\mathrm{QAB}(1,1)}^{(m)}(\lambda)
  &=B_{\mathrm{QBh}(2)}^{(m)}.
 \end{aligned}
 \label{eq:app-finite-zero-limits}
\end{equation}
The plotted bounds are maximized over $[-\pi,\pi]$, including these
limits. At $p=0.45$, $\kappa=6$ and $m=1,\ldots,10$, the optimized
QAB(1,1) values coincide numerically with QBh(2); all plotted bounds
remain below the actual variance in Eq.~\eqref{eq:simulation-bound-values}.

\bibliographystyle{apsrev4-2}
\bibliography{Refv3}

@book{Cramer1999,
  author   = {Cram{\'e}r, Harald},
  title    = {Mathematical Methods of Statistics},
  publisher= {Princeton University Press},
  year     = {1999},
  note     = {\url{https://press.princeton.edu/books/paperback/9780691005478/mathematical-methods-of-statistics}}
}

@incollection{Rao1992,
  author   = {Rao, C. Radhakrishna},
  editor   = {Kotz, Samuel and Johnson, Norman L.},
  title    = {Information and the accuracy attainable in the estimation of statistical parameters},
  booktitle= {Breakthroughs in Statistics: Foundations and Basic Theory},
  pages    = {235--247},
  publisher= {Springer},
  year     = {1992},
  doi      = {10.1007/978-1-4612-0919-5_16}
}

@article{Helstrom1969,
  author   = {Helstrom, Carl W.},
  title    = {Quantum detection and estimation theory},
  journal  = {J. Stat. Phys.},
  volume   = {1},
  pages    = {231--252},
  year     = {1969},
  doi      = {10.1007/BF01007479}
}

@article{Braunstein1994,
  author   = {Braunstein, Samuel L. and Caves, Carlton M.},
  title    = {Statistical distance and the geometry of quantum states},
  journal  = {Phys. Rev. Lett.},
  volume   = {72},
  number   = {22},
  pages    = {3439--3443},
  year     = {1994},
  doi      = {10.1103/PhysRevLett.72.3439}
}

@article{liu2021,
  author   = {Liu, Li-Zheng and Zhang, Yu-Zhe and Li, Zheng-Da and Zhang, Rui and Yin, Xu-Fei and Fei, Yue-Yang and Li, Li and Liu, Nai-Le and Xu, Feihu and Chen, Yu-Ao and others},
  title    = {Distributed quantum phase estimation with entangled photons},
  journal  = {Nat. Photonics},
  volume   = {15},
  number   = {2},
  pages    = {137--142},
  year     = {2021},
  doi      = {10.1038/s41566-020-00718-2}
}

@article{Kim2024Distributed,
  author   = {Kim, Dong-Hyun and Hong, Seongjin and Kim, Yong-Su and Kim, Yosep and Lee, Seung-Woo and Pooser, Raphael C. and Oh, Kyunghwan and Lee, Su-Yong and Lee, Changhyoup and Lim, Hyang-Tag},
  title    = {Distributed quantum sensing of multiple phases with fewer photons},
  journal  = {Nat. Commun.},
  volume   = {15},
  pages    = {266},
  year     = {2024},
  doi      = {10.1038/s41467-023-44204-z}
}

@article{liu2024,
  author   = {Liu, Biyao and Yang, Kui-Xing and Mao, Ya-Li and Feng, Lixin and Guo, Bixiang and Xu, Shufeng and Chen, Hu and Li, Zheng-Da and Fan, Jingyun},
  title    = {Experimental adaptive {Bayesian} estimation for a linear function of distributed phases in photonic quantum networks},
  journal  = {Optica},
  volume   = {11},
  number   = {10},
  pages    = {1419--1424},
  year     = {2024},
  doi      = {10.1364/OPTICA.532865}
}

@article{valeri2020,
  author   = {Valeri, Mauro and Polino, Emanuele and Poderini, Davide and Gianani, Ilaria and Corrielli, Giacomo and Crespi, Andrea and Osellame, Roberto and Spagnolo, Nicol{\`o} and Sciarrino, Fabio},
  title    = {Experimental adaptive Bayesian estimation of multiple phases with limited data},
  journal  = {npj Quantum Inf.},
  volume   = {6},
  number   = {1},
  pages    = {92},
  year     = {2020},
  doi      = {10.1038/s41534-020-00326-6}
}

@article{Minati2026Adaptive,
  author   = {Minati, Giorgio and Urbani, Enrico and Spagnolo, Nicol{\`o} and Cimini, Valeria and Sciarrino, Fabio},
  title    = {Multiparameter quantum-enhanced adaptive metrology with squeezed light},
  journal  = {Nat. Commun.},
  volume   = {17},
  pages    = {6418},
  year     = {2026},
  doi      = {10.1038/s41467-026-73069-1}
}

@article{polino2020,
  author   = {Polino, Emanuele and Valeri, Mauro and Spagnolo, Nicol{\`o} and Sciarrino, Fabio},
  title    = {Photonic quantum metrology},
  journal  = {AVS Quantum Sci.},
  volume   = {2},
  number   = {2},
  pages    = {024703},
  year     = {2020},
  doi      = {10.1116/5.0007577}
}

@article{Cimini2024Variational,
  author   = {Cimini, Valeria and Valeri, Mauro and Piacentini, Simone and Ceccarelli, Francesco and Corrielli, Giacomo and Osellame, Roberto and Spagnolo, Nicol{\`o} and Sciarrino, Fabio},
  title    = {Variational quantum algorithm for experimental photonic multiparameter estimation},
  journal  = {npj Quantum Inf.},
  volume   = {10},
  pages    = {26},
  year     = {2024},
  doi      = {10.1038/s41534-024-00821-0}
}

@article{Giovannetti2006,
  author   = {Giovannetti, Vittorio and Lloyd, Seth and Maccone, Lorenzo},
  title    = {Quantum metrology},
  journal  = {Phys. Rev. Lett.},
  volume   = {96},
  number   = {1},
  pages    = {010401},
  year     = {2006},
  doi      = {10.1103/PhysRevLett.96.010401}
}

@article{Degen2017,
  author   = {Degen, C. L. and Reinhard, F. and Cappellaro, P.},
  title    = {Quantum sensing},
  journal  = {Rev. Mod. Phys.},
  volume   = {89},
  number   = {3},
  pages    = {035002},
  year     = {2017},
  doi      = {10.1103/RevModPhys.89.035002}
}

@article{Pezze2018,
  author   = {Pezz{\`e}, Luca and Smerzi, Augusto and Oberthaler, Markus K. and Schmied, Roman and Treutlein, Philipp},
  title    = {Quantum metrology with nonclassical states of atomic ensembles},
  journal  = {Rev. Mod. Phys.},
  volume   = {90},
  number   = {3},
  pages    = {035005},
  year     = {2018},
  doi      = {10.1103/RevModPhys.90.035005}
}

@article{Montenegro2025Review,
  author   = {Montenegro, Victor and Mukhopadhyay, Chiranjib and Yousefjani, Rozhin and Sarkar, Saubhik and Mishra, Utkarsh and Paris, Matteo G. A. and Bayat, Abolfazl},
  title    = {Review: Quantum metrology and sensing with many-body systems},
  journal  = {Phys. Rep.},
  volume   = {1134},
  pages    = {1--62},
  year     = {2025},
  doi      = {10.1016/j.physrep.2025.05.005}
}

@article{Antonella2016,
  author  = {A. De Pasquale and D. Rosmini and R. Fazio  and V. Giovannetti
             },
  title   = {Local quantum thermal susceptibility},
  journal = {Nat. Commun.},
  volume  = {7},
  pages   = {12782},
  year    = {2016},
  doi     = {10.1038/ncomms12782}
}

@article{Albarelli2026Measurement,
  author   = {Albarelli, Francesco and Branford, Dominic and Rubio, Jes{\'u}s},
  title    = {Measurement incompatibility in {Bayesian} multiparameter quantum estimation},
  journal  = {Quantum},
  volume   = {10},
  pages    = {2192},
  year     = {2026},
  doi      = {10.22331/q-2026-08-13-2192}
}

@article{DiFresco2022Multiparameter,
  author   = {Di Fresco, Giovanni and Spagnolo, Bernardo and Valenti, Davide and Carollo, Angelo},
  title    = {Multiparameter quantum critical metrology},
  journal  = {SciPost Phys.},
  volume   = {13},
  pages    = {077},
  year     = {2022},
  doi      = {10.21468/SciPostPhys.13.4.077}
}

@article{DiFresco2025CriticalReview,
  author   = {Di Fresco, Giovanni and Spagnolo, Bernardo and Valenti, Davide and Carollo, Angelo},
  title    = {Multi-parameter estimation in critical quantum metrology},
  journal  = {Chaos Solitons Fractals},
  volume   = {199},
  pages    = {116913},
  year     = {2025},
  doi      = {10.1016/j.chaos.2025.116913}
}

@article{DiFresco2024Measurement,
  author   = {Di Fresco, Giovanni and Spagnolo, Bernardo and Valenti, Davide and Carollo, Angelo},
  title    = {Metrology and multipartite entanglement in measurement-induced phase transition},
  journal  = {Quantum},
  volume   = {8},
  pages    = {1326},
  year     = {2024},
  doi      = {10.22331/q-2024-04-30-1326}
}

@article{Carollo2018FermionicSLD,
  author   = {Carollo, Angelo and Spagnolo, Bernardo and Valenti, Davide},
  title    = {Symmetric Logarithmic Derivative of Fermionic Gaussian States},
  journal  = {Entropy},
  volume   = {20},
  number   = {7},
  pages    = {485},
  year     = {2018},
  doi      = {10.3390/e20070485}
}

@article{Bhattacharyya1946,
  author   = {A. Bhattacharyya},
  title    = {On Some Analogues of the Amount of Information and Their Use in Statistical Estimation},
  journal  = {Sankhy\={a}},
  volume   = {8},
  pages    = {1--14},
  year     = {1946},
  note     = {\url{https://www.jstor.org/stable/25047921}}
}

@article{Hammersley1950,
  author   = {J. M. Hammersley},
  title    = {On estimating restricted parameters},
  journal  = {J. R. Stat. Soc. Ser. B},
  volume   = {12},
  number   = {2},
  pages    = {192--229},
  year     = {1950},
  doi      = {10.1111/j.2517-6161.1950.tb00056.x}
}

@article{Chapman1951,
  author   = {D. G. Chapman and H. Robbins},
  title    = {Minimum variance estimation without regularity assumptions},
  journal  = {Ann. Math. Statist.},
  volume   = {22},
  pages    = {581},
  year     = {1951},
  doi      = {10.1214/aoms/1177729548}
}

@article{Abel1993,
  author   = {J. S. Abel},
  title    = {A bound on mean square estimate error},
  journal  = {IEEE Trans. Inf. Theory},
  volume   = {39},
  pages    = {1675},
  year     = {1993},
  doi      = {10.1109/18.259655}
}

@article{Barankin1949,
  author   = {E. W. Barankin},
  title    = {Locally best unbiased estimates},
  journal  = {Ann. Math. Statist.},
  volume   = {20},
  pages    = {477},
  year     = {1949},
  doi      = {10.1214/aoms/1177729943}
}

@article{Gessner2023,
  author   = {Gessner, Manuel and Smerzi, Augusto},
  title    = {Hierarchies of frequentist bounds for quantum metrology: From {Cram{\'e}r--Rao} to {Barankin}},
  journal  = {Phys. Rev. Lett.},
  volume   = {130},
  number   = {26},
  pages    = {260801},
  year     = {2023},
  doi      = {10.1103/PhysRevLett.130.260801}
}

@article{Gebhart2024FewMeasurements,
  author   = {Gebhart, Valentin and Gessner, Manuel and Smerzi, Augusto},
  title    = {Fundamental bounds for parameter estimation with few measurements},
  journal  = {Phys. Rev. Research},
  volume   = {6},
  number   = {4},
  pages    = {043261},
  year     = {2024},
  doi      = {10.1103/PhysRevResearch.6.043261}
}

@article{Hervas2025HigherOrder,
  author   = {Hervas, J. R. and Goldberg, A. Z. and Sanz, A. S. and Hradil, Z. and {\v{R}}eh{\'a}{\v{c}}ek, J. and S{\'a}nchez-Soto, L. L.},
  title    = {Higher-order corrections to the quantum {Cram{\'e}r--Rao} bound},
  journal  = {Phys. Rev. A},
  volume   = {112},
  number   = {2},
  pages    = {022426},
  year     = {2025},
  doi      = {10.1103/ltz3-163t}
}

@article{Meyer2025FiniteSample,
  author   = {Meyer, Johannes Jakob and Khatri, Sumeet and Stilck Fran{\c{c}}a, Daniel and Eisert, Jens and Faist, Philippe},
  title    = {Quantum metrology in the finite-sample regime},
  journal  = {PRX Quantum},
  volume   = {6},
  number   = {3},
  pages    = {030336},
  year     = {2025},
  doi      = {10.1103/qbn1-p6bq}
}

@article{BelfieldBrody2021,
  author   = {Belfield, Alex J. and Brody, Dorje C.},
  title    = {Higher-order uncertainty bounds for mixed states},
  journal  = {J. Phys. A: Math. Theor.},
  volume   = {54},
  number   = {43},
  pages    = {435302},
  year     = {2021},
  doi      = {10.1088/1751-8121/ac2717}
}

@article{stoica1998,
  author   = {Stoica, Petre and Ng, Boon Chong},
  title    = {On the Cram{\'e}r-Rao bound under parametric constraints},
  journal  = {IEEE Signal Process. Lett.},
  volume   = {5},
  number   = {7},
  pages    = {177--179},
  year     = {1998},
  doi      = {10.1109/97.700921}
}

@article{Ben-Haim2010,
  author   = {Ben-Haim, Zvika and Eldar, Yonina C},
  title    = {The Cram{\'e}r-Rao bound for estimating a sparse parameter vector},
  journal  = {IEEE Trans. Signal Process.},
  volume   = {58},
  number   = {6},
  pages    = {3384--3389},
  year     = {2010},
  doi      = {10.1109/TSP.2010.2045423}
}

@article{cianchi2013,
  author   = {Cianchi, Andrea and Lutwak, Erwin and Yang, Deane and Zhang, Gaoyong},
  title    = {A unified approach to Cram{\'e}r--Rao inequalities},
  journal  = {IEEE Trans. Inf. Theory},
  volume   = {60},
  number   = {1},
  pages    = {643--650},
  year     = {2014},
  doi      = {10.1109/TIT.2013.2284498}
}

@article{ap2006,
  author   = {J. B. Altepeter and E. R. Jeffrey and P. G. Kwiat},
  title    = {Photonic state tomography},
  journal  = {Adv. At. Mol. Opt. Phys.},
  volume   = {52},
  pages    = {105},
  year     = {2005},
  doi      = {10.1016/S1049-250X(05)52003-2}
}

@article{HHH,
  author   = {R. Horodecki and P. Horodecki and M. Horodecki and K. Horodecki},
  title    = {Quantum entanglement},
  journal  = {Rev. Mod. Phys.},
  volume   = {81},
  pages    = {865},
  year     = {2009},
  doi      = {10.1103/RevModPhys.81.865}
}

\end{document}